\documentclass{article} 
\usepackage{nips15submit_e,times}
\usepackage{hyperref}
\usepackage{url}

\usepackage{amsmath,amssymb,amsthm, widetext, xspace}
\usepackage{subfigure,algorithmic,algorithm, dsfont}
\usepackage{epsfig}
\usepackage{graphicx}
\usepackage[numbers]{natbib}
\newcommand{\var}{\mathrm{var}}

\newcommand{\ber}{\mathrm{Ber}}
\newcommand{\figref}[1]{Fig.~\ref{fig:#1}}

\newcommand{\thmref}[1]{Theorem~\ref{theorem:#1}}
\newcommand{\figuredir}{./}
\newcommand{\appref}[1]{Appendix~\ref{sec:#1}}
\newcommand{\p}{{\mu}_p}
\newcommand{\bp}{{\hat{\mu}_p}}
\newcommand{\hp}{{\hat{\mu}_p}}
\newcommand{\tp}{{\hat{\mu}_p}}
\newcommand{\ntrad}{{{n}_a^*}}
\newcommand{\nhyb}{{{n}_b}}
\newcommand{\hhyb}{{{n}_a}}
\newcommand{\ntrans}{{n}_b}

\newcommand{\na}{{n_a}}
\newcommand{\nb}{{n_b}}

\newcommand{\cc}{{c_c}}
\newcommand{\ca}{{c_a}}
\newcommand{\cb}{{c_b}}
\newcommand{\ftexton}{$f_b^{\textrm{(texton)}}$\xspace}
\newcommand{\fCNN}{$f_b^{\textrm{(CNN)}}$\xspace}
\newcommand{\hdotsc}{, \hdots ,}

\newtheorem{Lemma}{Lemma}
\newtheorem{Theorem}{Theorem}

\title{Random Sampling in an Age of Automation: Minimizing Expenditures through Balanced Collection and Annotation}

\author{
Oscar Beijbom \\
UC Berkeley\\
\texttt{obeijbom@berkeley.edu}\\
}

\nipsfinalcopy 

\begin{document}

\maketitle

\begin{abstract}
Methods for automated collection and annotation are changing the cost-structures of sampling surveys for a wide range of applications. Digital samples in the form of images or audio recordings can be collected rapidly, and annotated by computer programs or crowd workers. We consider the problem of estimating a population mean under these new cost-structures, and propose a Hybrid-Offset sampling design. This design utilizes two annotators: a primary, which is accurate but costly (e.g. a human expert) and an auxiliary which is noisy but cheap (e.g. a computer program), in order to minimize total sampling expenditures. Our analysis gives necessary conditions for the Hybrid-Offset design and specifies optimal sample sizes for both annotators. Simulations on data from a coral reef survey program indicate that the Hybrid-Offset design outperforms several alternative sampling designs. In particular, sampling expenditures are reduced $50\%$ compared to the Conventional design currently deployed by the coral ecologists. 
\end{abstract}

\section{Introduction}
Using random sampling to estimate the mean of a population is a fundamentally important method to the sciences and society at large, and has been studied extensively~\cite{sampling, deming1966some}. Deployment of any random sampling design requires collection of some number of observations sampled randomly from nature. In the ecological sciences, this was traditionally done \emph{in situ} by an expert. Recently, advances in robotics, sensor technology, digital storage, and information technology have enabled rapid collection of samples in digital format, such as images~\cite{rango2009unmanned, hunt2010acquisition} or audio~\cite{johnson2003digital}. The popularity of digital sample collection can be attributed to three key factors: it creates a permanent record; it can be done cheaply using automated sampling vehicles or non-expert personnel; and it is generally fast. However, such samples (e.g. a photoquadrat of the forest floor) typically require annotation by an expert in order to reveal the desired quantity of interest (e.g. a count of insects). Such annotation work can be slow, tedious, expensive, and prone to error~\cite{macleod2010time, ninio2003estimating}. 

Concurrent with the development of automated collection methods, advances in computer-vision and computer-audition have enabled automation of said annotation work. Such methods often rely on machine learning where expert annotated archived data sets are utilized to train automated annotators. Automation is a compelling low-cost alternative to expert annotations, but it's generally less reliable and may be biased~\cite{beijbom2012automated, hand2006classifier, torralba2011unbiased}. This is particularly problematic if the probability density of the archived data differs from the density of the data to be sampled~\cite{pan2010survey}.

Crowdsourcing offers another low-cost alternative to expert annotation for e.g. document or image annotation~\cite{welinder, platemate}. Crowdsourced annotations can be noisy, and much work has been devoted to improving the quality of such annotations. This is generally done either by carefully designing the tasks given to the crowd workers~\cite{platemate}, or by collecting multiple crowd annotations for the same sample and then modeling, and compensating for, the annotation errors~\cite{welinder}.

We consider the problem of estimating a population mean under these new cost-structures of data collection and annotation. This is formalized as follows: Given a procedure for collecting random samples, $x_i \in {\cal X}$  (e.g. images or audio recordings), each with an associated quantity of interest (value), $y_i \in {\cal Y} \subseteq [0,  1]$; and two annotators: a primary, $f_a(x_i) \in {\cal Y}$ which is accurate but costly (e.g. a human expert), and an auxiliary, $f_b(x_i)\in {\cal Y}$ which is cheap but noisy (e.g. a computer program, lay-person, or crowd worker).  Our goal is to derive a sampling design that achieves unbiased estimates of the population mean ($E[y_i]$) at a target error and confidence, while minimizing total cost of collection and annotation. In particular, we investigate the optimal balance between the number of samples annotated by the primary and auxiliary annotators. This work is, to the best of our knowledge, the first to consider this problem.

A key challenge is to define a procedure that can correct for the potential bias of the auxiliary annotator. This is difficult as we cannot assume any prior knowledge of the underlying probability density from which the samples are drawn. Indeed, if this density was known, the population mean could be estimated directly, making the sampling work unnecessary. If the auxiliary annotator is based on machine learning and trained on archived data with a different probability density, the problem of transfer learning arises for which the generalization bounds of statistical learning theory generally do not apply~\cite{pan2010survey}. Methods for bias-correction have been proposed independently by ~\cite{solow2001estimating, hopkins2010method}, that do not require knowledge of the full underlying probability density (of the sampled data) but only of the conditional probabilities of a label given a sample. However, as we shall see, this information may not always be at hand. 

The key contribution of this work is an analysis of a Hybrid-Offset design that directly models the offset (bias) of the auxiliary annotator. It is ``Hybrid'' because it requires a subset of the samples to be annotated by both annotators, and is unbiased if the samples are independent and identically distributed, which they are by construction under random sampling. As demonstrated by simulations on coral reef survey data, the Hybrid-Offset design is cost-effective and robust. In particular, it outperforms a Hybrid-Ratio design which utilizes the ratio-estimator commonly used in the sampling literature~\cite{sampling,royall1981empirical}. It also outperforms several designs that only rely on one of the annotators, including the design currently used by the coral ecologists. We believe that the Hybrid-Offset can be widely utilized, in particular for ecological surveys relying on digital samples~\cite{willis2000baited, olson2007submersible, rango2009unmanned, archer1988autogenic, johnson2003digital}. Other contributions of this work are: (1) an analysis of the bias-correction method of \cite{solow2001estimating, hopkins2010method} in the context of random sampling; and (2) an improved machine-learning method, based on convolutional neural networks, for annotation of coral reef survey images. 

\subsection{Related work}
Our work is most closely related to the literature on survey sampling with auxiliary data~\cite{sampling,royall1981empirical}. In that context, ``auxiliary data'' is typically not a direct estimate of the variable of interest but some other, related quantity. For example, if the variable of interest, $y_i$ is the number of animals per plot, the auxiliary data can be the plot area, vegetation type, or plot elevation. In this literature, auxiliary data is typically incorporated using a ratio estimator which can reduce the estimation errors if there is an approximately linear relationship between the auxiliary data and the variable of interest~\cite{sampling}. However, the ratio estimator is design-biased, and analysis of the estimator variance typically assumes that auxiliary data is available for the whole population~\cite{sampling,royall1981empirical}. In this work, in contrast, we do not assume that the auxiliary data is available for all samples, but that this data can be acquired at a cost by collecting additional samples. This additional cost is then taken into account when deriving optimal sampling sizes. Another key difference is that we do not use the ratio estimator but an offset estimator which directly estimates the bias of the auxiliary annotator. The offset estimator is design-unbiased and allows for a more straight-forward analysis.

Another line of related work utilizes stratified random sampling~\cite{bennett2010online}, importance sampling~\cite{sawade2010active} or generative models of the classifier score distribution~\cite{welinder2013lazy}, to achieve cost-effective estimates of classifier performance on new data. 
The work of Garnett et al.~\cite{garnett2012bayesian} is particularly relevant, and investigates methods for active selection of samples in order to estimate class proportions. However, these methods all operate on a fixed set of samples. In contrast, we include the sample collection in our model which enables a joint minimization of annotation and collection costs.

Our work is also related to active learning and transfer learning. It is related, in particular, to recent work on active transfer learning where labels are queried to optimize classifier performance in a target domain~\cite{wang2014active}. A key difference between that work and ours is that active learning methods optimize the labeling effort to create the best \emph{classifier} (which then can, presumably, be used to label more data and in order to estimate the desired data-products). In contrast, we directly optimize the labeling effort to derive the desired \emph{data-product} (i.e. the population mean).

\section{Preliminaries}
\subsection{Problem Setup}
We denote by $\p \equiv E[y_i]$ and $\sigma_p^2 \equiv \var[y_i]$ the first and second moments of the (unknown) probability density function of the values. The values are sampled randomly from nature, and are therefore i.i.d. We further denote by $f_a: {\cal X} \rightarrow {\cal Y}$ the primary, and by $f_b: {\cal X} \rightarrow {\cal Y}$ the auxiliary annotator, $\epsilon_{a, i} \equiv f_a(x_i) - y_i$ the error of $f_a$ on sample $i$, $\mu_a \equiv E[\epsilon_{a, i}]$ and $\sigma_a^2 \equiv \var[\epsilon_{a, i}]$. Similarly, $\mu_b \equiv E[\epsilon_{b, i}]$ and $\sigma_b^2 \equiv \var[\epsilon_{b, i}]$. Note that we do \emph{not} make any assumptions on the underlying probability densities of the sample values or annotator errors.

We denote by $\na$ and $\nb$ the number of samples annotated by $f_a$ and $f_b$, respectively. The number of collected samples is given by $\max(n_a, n_b)$ since samples needs to be annotated to provide any information, and conversely, needs to be collected in order to be annotated. We denote by $\cc, \ca$ and $\cb$ the cost per sample for collection, annotation by $f_a$ and annotation by $f_b$, respectively. The `accurate and expensive' characteristics of $f_a$ are operationalized by letting $\sigma_a^2 < \sigma_b^2$ and $\ca > \cb$. 

We can now precisely state our goal: Given costs $\cc, \ca$ and $\cb$, and two annotators, $f_a$ and $f_b$, derive a sampling design that estimates the population mean, $\p$ by defining the number of annotated samples ($\na$ and $\nb$), so that $E[\hat{\mu}_p] = \p$ and $\textrm{Pr}(|\hat{\mu}_p - \p| > d) < \delta$, for a target error, $d$ and confidence, $\delta$. The utility of the sampling design is evaluated by the Total Sampling Cost (TSC), $b$:
\begin{align}
b(\na, \nb) = c_a \na + c_b \nb + \max(\na, \nb) c_c.
\end{align}

We make three assumptions. First, we assume that the number of collected samples is small in comparison with the total size of the population which allows us to omit the finite-population correction factor~\cite{sampling}. Second, we assume that the primary annotator, $f_a$ is unbiased, i.e. $\mu_a = 0$, and that the correlation between $\epsilon_{a,i}$ and $y_i$ is negligible. Third, because the two annotators are independent entities, we assume zero correlation between the annotator errors $\epsilon_{a, i}$ and $\epsilon_{b, i}$. However, we do \emph{not} make any assumptions on the correlation between the auxiliary annotation errors $\epsilon_{b,i}$ and the sample value $y_i$, which may be large. All proofs are in the Appendix.

\subsection{Conventional design}
We denote by `conventional', a sampling design where all collected samples, $x_i$ are annotated by the primary annotator $f_a$, i.e. $n_b = 0$. In such design, an unbiased estimator of $\p$ is given by
\begin{align}
\bp = \frac{1}{\na} \sum_{i=1}^{\na} f_a(x_i), \label{eq:pbreve}
\end{align}
with variance
\begin{align}
\var[\bp] =  \frac{1}{\na} (\sigma_p^2 + \sigma_a^2). \label{eq:pvar_breve}
\end{align}
The variance ($\sigma_p^2$ + $\sigma_a^2$) is often unknown, and must be estimated by the sample variance of $f_a(x_1) \hdotsc f_a(x_{\na})$. The sample size, $\na$ needs to be large enough to ensure that $\textrm{Pr}(|\bp - \p| > d) < \delta$, for a target error $d$ and confidence $\delta$. From the Central Limit Theorem, this is satisfied when
\begin{align}
\zeta_{\delta} \sqrt{ \var [\bp]} \leq d, \label{eq:varreq}
\end{align}
where $\zeta_{\delta}$ is the upper $1-\delta/2$ point on on the standard normal distribution curve~\cite{sampling}. The target sample size is given by inserting \eqref{eq:pvar_breve} into \eqref{eq:varreq} yielding
\begin{align}
\ntrad = \frac{\zeta_{\delta}^2}{d^2} (\sigma_p^2 + \sigma_a^2), \label{eq:ntrad}
\end{align}
for a TSC: $(\cc + c_a)\ntrad$.

\section{Hybrid-Offset design}
Now consider a hybrid design where $\nb \geq \ntrad$ samples are collected and annotated by the auxiliary annotator $f_b$, and where a subset $\na \leq \nb$ is also annotated by the primary annotator $f_a$. An offset estimator of $\p$ under this design is given by 
\begin{align}
\hp =  \frac{1}{\nb} \sum_{i = 1}^\nb f_b(x_i) - \hat{\mu}_b. \label{eq:phat}
\end{align}
The offset estimator is unbiased and an unbiased estimate of $\mu_b$ is given by
\begin{align}
\hat{\mu}_b = \frac{1}{\na} \sum_{i=1}^{\na} f_b(x_i) - f_a(x_i). \label{eq:bhat}
\end{align}
The variance of $\hp$ is given by
\begin{align}
\var[\hp] = \frac{1}{\nb} (\sigma_p^2 - \sigma_b^2 ) + \frac{1}{\na} (\sigma_a^2  + \sigma_b^2), \label{eq:pvar}
\end{align}
and notably does \emph{not} depend on the covariance between $y_i$ and $\epsilon_{b,i}$. This follows directly from the derivations of \eqref{eq:pvar}, which are provided in \appref{offset_variance}. We denote by ``Hybrid-Offset'' a design that balances $n_a$ and $n_b$ to minimize the TSC.

If the costs are such that a large number of samples, $\nb$ can be collected and annotated by the auxiliary annotator, the magnitude of the first term in \eqref{eq:pvar} becomes small and the sampling error depends mainly on the auxiliary annotation error, $\sigma_b^2$. In contrast, the conventional sampling design depends mainly on the data variance, $\sigma_p^2$ \eqref{eq:pvar_breve}. This is compelling because while the data variance is a fixed constant of nature, the auxiliary annotation errors depend on the choice and quality of the auxiliary annotator, which we control. It also leads to our first result:
\begin{Theorem}
\label{theorem:1}
For any fixed $\nb > \na$, the variance of the Hybrid-Offset estimator~\eqref{eq:pvar} is smaller than the variance of the Conventional design~\eqref{eq:pvar_breve} if and only if $\sigma_b^2 < \sigma_p^2$.
\end{Theorem}

Theorem~\ref{theorem:1} implies that the uncertainty introduced by $f_b$ can be compensated for by using more samples if and only if $\sigma_b^2 < \sigma_p^2$. However, the additional collection of samples is only economical for certain cost functions, and should in the general case be determined by comparing the TSC of the two designs. To determine the TSC of the Hybrid-Offset design we begin by deriving optimal sample sizes $\nhyb$ and $\hhyb$. By combining \eqref{eq:pvar} and \eqref{eq:varreq}, and solving for equality, the following trade-off between $\nb$ and $\na$ is derived
\begin{align}
\na = \frac{\sigma_b^2 + \sigma_a^2}{\frac{d^2}{\zeta_{\delta}^2}  - \frac{1}{\nb}(\sigma_p^2 - \sigma_b^2)}. \label{eq:hn}
\end{align}
Example trade-off curves demonstrate how the amount of primary annotations can be reduced by increasing the amount of auxiliary annotations (\figref{hn_levels}A). Note that if $\nb = \ntrad$, it follows from \eqref{eq:hn} that $\na = \ntrad$, and the Hybrid-Offset design reduces to the Conventional design. The optimal operating point along the $\nb, \na$ trade-off curve can be derived by minimizing the TSC. Using \eqref{eq:hn} to eliminate $\na$, the TSC becomes
\begin{align}
b(\na, \nb) = (\cc + c_b) \left( \nb + \frac{k(\sigma_b^2 + \sigma_a^2)}{\frac{d^2}{\zeta_{\delta}^2}  - \frac{1}{\nb}(\sigma_p^2 - \sigma_b^2)}  \right), \label{eq:tcd_offset}
\end{align}
where $k = \frac{c_a}{\cc+c_b}$ is the relative cost of $f_a$. The optimal sample size, $\nhyb$ is given by minimizing \eqref{eq:tcd_offset} under the constraint that $\nb \geq \ntrad$. This yields the following theorem:
\begin{Theorem}
\label{theorem:2}
Under a Hybrid-Offset design \eqref{eq:phat}, the optimal auxiliary annotation size is
\begin{align}
\nhyb = \max \left[ \frac{\zeta_{\delta}^2}{d^2} \left( \sigma_p^2 - \sigma_b^2 + \sqrt{k (\sigma_b^2+ \sigma_a^2) ( \sigma_p^2 - \sigma_b^2 )}~ \right), \ntrad \right] \label{eq:nopt}
\end{align}
The corresponding primary annotation size is given by \eqref{eq:hn}.
\end{Theorem}
Using the optimal sample sizes, the TSC of Hybrid-Offset sampling can be calculated from \eqref{eq:tcd_offset} and compared to the TSC of Conventional sampling in order to determine the most cost-effective design. For the important special case where $c_b = 0$ (which can occur e.g. if $f_b$ is a computer algorithm) the following theorem applies:
\begin{Theorem}
\label{theorem:3}
If $c_b = 0$, the TSC of Hybrid-Offset sampling is smaller than the TSC of Conventional sampling if and only if $k > \sigma_{\Delta}$, where $k = \frac{c_a}{\cc+c_b}$ and $\sigma_{\Delta} = \frac{\sigma_a^2 + \sigma_b^2}{\sigma_p^2 - \sigma_b^2}$. 
\end{Theorem}
For example, if the primary annotation errors are zero ($\sigma_a^2=0$), and the auxiliary errors are half as large as the data variance ($\sigma_b^2 = \sigma_p^2 / 2$), then $\sigma_{\Delta} = 1$. This means that Hybrid-Offset sampling is cheaper than Conventional sampling if $k > 1$, which occurs if the cost of collection is smaller than the cost of primary annotation. The difference in TSC between the two sampling designs is shown in \figref{tcd}B for various values of $k, \sigma_p$ and $\sigma_b$.

Hybrid-Offset sample sizes can also be derived directly from a target TSC, $b$, and costs, $c_a, c_b$ and $c_c$ by minimizing \eqref{eq:pvar} under the TSC constraint: $b = \na c_a + \nb (c_b + c_c)$. This yields:
\begin{align}
n_b =  \frac{b \sqrt{c_a(c_b+c_c) \sigma_{\Delta}}}{c_a(c_b+c_c) \sigma_{\Delta} - (c_b+c_c)^2}  - \frac{b}{c_a\sigma_{\Delta} - (c_b+c_c)}, \label{eq:sample_sizes_from_budget}
\end{align}
where as previous, $\sigma_{\Delta} = \frac{\sigma_a^2 + \sigma_b^2}{\sigma_p^2 - \sigma_b^2}$. The corresponding $n_a$ is given by the TSC constraint.

\begin{figure}[t]
   \begin{center}
    \includegraphics[height=0.29\linewidth]{\figuredir 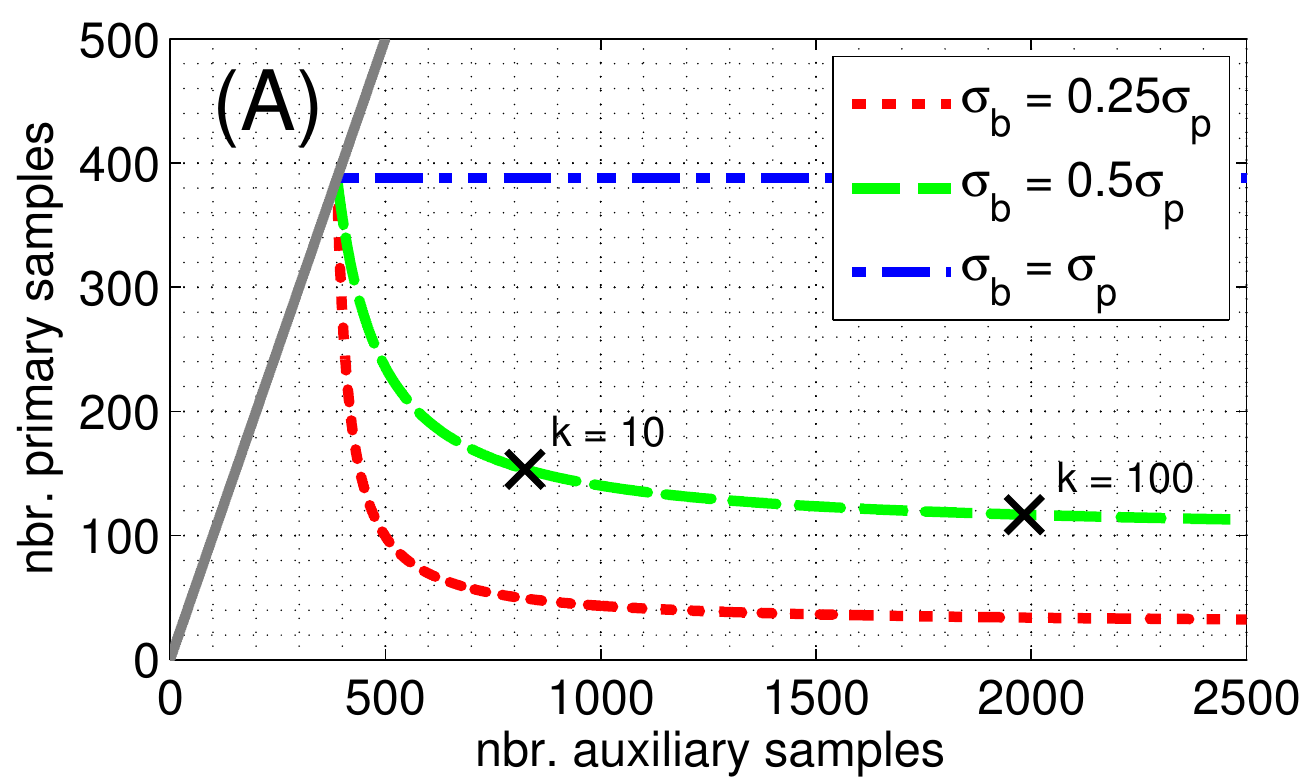}
    \includegraphics[height=0.29\linewidth]{\figuredir 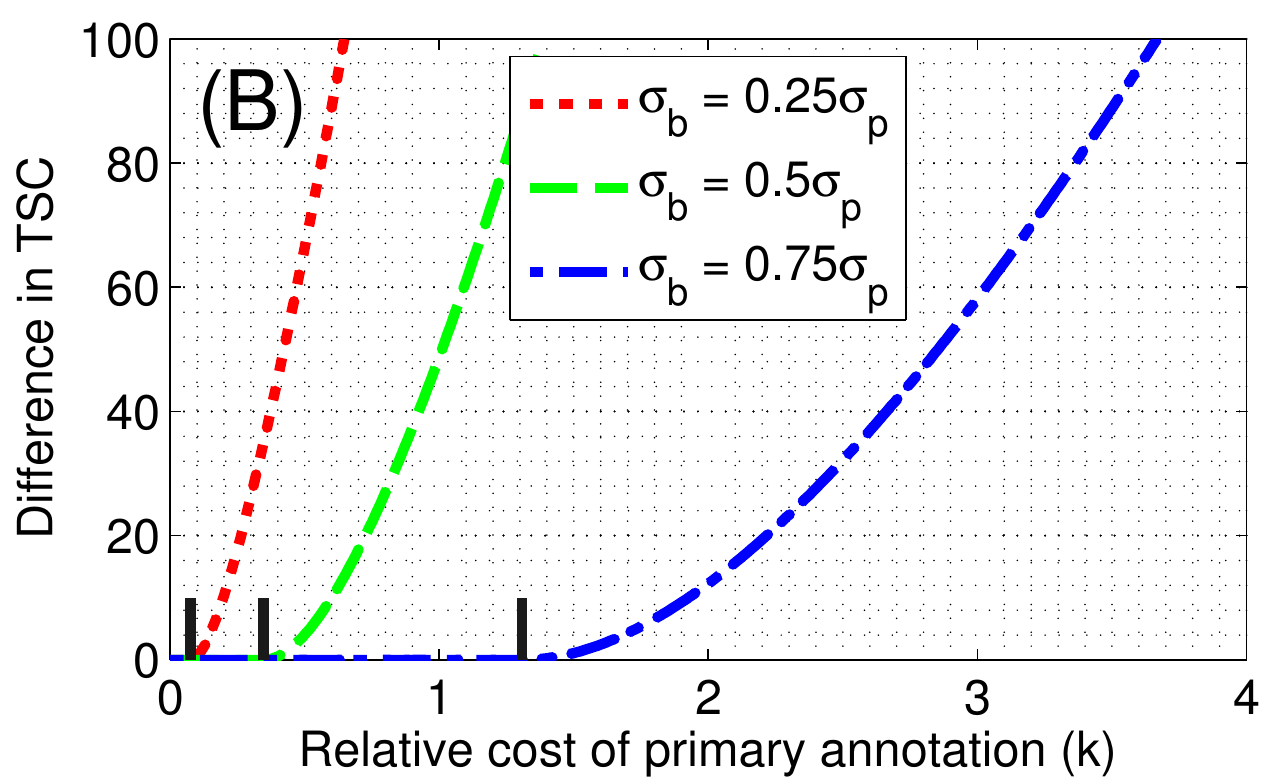}
    \end{center} \vspace{-4mm}
   \caption{\textbf{(A)} Amount of auxiliary ($\nb$), and primary ($\na$) annotation required to achieve error $d \leq 0.02$ at $95\%$ confidence for $\sigma_p = 0.2$, $\sigma_a = 0.02$, and $\sigma_b = \{0.25\sigma_p, 0.5\sigma_p, \sigma_p\}$ under the Hybrid-Offset sampling design. Solid gray line indicates $\nb = \na$. Optimal operating points \eqref{eq:nopt}, for relative cost of primary annotation, $k = 10$ and $k = 100$ are marked with X on the $\sigma_b = 0.5\sigma_p$ curve. \textbf{(B)} Difference in Total Sampling Cost ($\text{TSC}_{\text{Conventional}} - \text{TSC}_{\text{offset}}$) for $\sigma_p = 0.2$, $\sigma_a = 0.02$, $\sigma_b = \{0.25\sigma_p, 0.5\sigma_p, 0.75\sigma_p\}$, $c_b = 0$ and $c_c = 1$, as a function of the relative cost of primary annotation, $k = c_a / (c_c + c_b)$. The threshold costs, $\sigma_{\Delta}$ are marked with black stumps. If $k < \sigma_{\Delta}$ the sampling designs and costs are identical.} \vspace{-2mm}
   \label{fig:hn_levels}
      \label{fig:tcd}
\end{figure}

\section{Experiments}
The proposed method is discussed in the context of an annual coral reef survey performed by the Moorea Coral Reef Long Term Ecological Research (MCR-LTER) program (http://mcr.lternet.edu). The program surveys six sites across the island of Moorea in French Polynesia. At each site, three habitats are surveyed: the fringing-reef and two habitats on the outer reef at 10 and 17 meter depth respectively, for a total of 18 sampling ``units''. In each unit the goal, as dictated by the ecologists, is to estimate the percent cover of key benthic substrates, such as coral and algae. These data provide important information about the ecology when compared across sites, habitats and years. To estimate the percent cover for each unit, ecologists capture photographs (\emph{in-situ} by a research-diver) at random locations along five line-transects at each site. The photographs are then annotated in order to estimate the percent cover for each photograph. This is done through random point sampling in which the substrate is identified at $200$ random point locations in each photograph~\cite{pante2012getting}. This procedure of using random sampling to annotate each collected random sample is commonly referred to as two-stage sampling~\cite{sampling, deming1966some}. 

For the purpose of this discussion we will focus on the estimation quality of percent coral cover for each unit, and investigate the effect of using an Hybrid-Offset design in place of the Conventional designs currently in use. For the simulations we use the Moorea Labeled Corals dataset, which is publicly available\footnote{ \label{note1} http://vision.ucsd.edu/content/moorea-labeled-corals} and contains the full-resolution images and annotations from the LTER-MCR surveys conducted 2008 and 2009. We use the data from 2008 to train the auxiliary annotators, and estimate the sampling parameters, and the data from 2009 to run the sampling simulations.

As described above, each photograph, $x_i$, is a random sample with a corresponding coral cover $y_i \in [0, 1]$, and $n_a = 40$ in each unit. The coral cover estimated by the expert annotator, $f_a(x_i)$ is highly accurate~\cite{ninio2003estimating}, and we therefore use $\sigma_a^2 = 0$ in the simulations below. We do not account for approximation errors introduced by the second-stage (point-annotation) method, as this have been shown to have very limited effect on the final mean estimator~\cite{deming1966some,pante2012getting}. Manual annotation requires $\sim 10$ minutes per image to complete, while collection is quicker, with the $40$ samples in a sampling-unit captured in a single 40 minute dive. With these parameters, the TSC for each unit is approximately $40(10+1)$ minutes, or 7.3 person-hours. 

\textbf{Auxiliary annotators:} We use two auxiliary annotators, the ``texton''-based classifier proposed by~\cite{beijbom2012automated} which is publicly available, and a novel annotator based on convolutional neural networks (CNN)~\cite{girshick2014rich}. Both of these methods operate on $p \times p$ pixel images patches, and are denoted $f_{b^*}^{(\Psi)}(x_{i,j}): \mathbb{R}^{p \times p \times 3} \rightarrow \{0, 1\}$, for $\Psi \in \{ \textrm{texton}, \textrm{CNN} \}$, where $x_{i,j}$ is an patch from image $i$ around random point-location $j$, and where output 1 indicates `coral', and 0 `other'. The output of the auxiliary annotator, $f_b^{(\Psi)}(x_i)$ is given by averaging the point-classifications for the $200$ random points in each image:
\begin{align}
f_b^{(\Psi)}(x_i) = \frac{1}{200} \sum_{j=1}^{200} f_{b^*}^{(\Psi)}(x_{i,j}). \label{eq:fbstar}
\end{align}
To develop the CNN based point-classifier, \fCNN, we adopted a 16 layer CNN model developed for image classification~\cite{simonyan2014very}. The VGG16 model is publicly available\footnote{https://github.com/BVLC/caffe/wiki/Model-Zoo} and operates on $226 \times 226$ RGB images. To fine-tune this model for coral classification, we cropped $226 \times 226$ patches from the 2008 images at each of the 200 annotated point locations. These cropped patches were used together with mirrored and rotated (by $90, 180$ and $270$ degrees) versions as training-data to fine-tune the weights of the VGG16 model largely following the procedure proposed by~\cite{girshick2014rich}. Classification was performed by cropping patches from the test images, and propagating them through the network.

\textbf{Cost analysis:} Using the data from 2008, the sample variance of the percent coral cover $\hat{\sigma}_p$ was $0.16 \pm 0.1$ (mean $\pm$ SD, $n=18$) across the units, meaning that for an average unit, a 95\% confidence interval of mean coral cover, from the $n_a=40$ samples, is 5.8\%. By cross-validation on the data from 2008, the auxiliary annotator errors of \ftexton, $\hat{\sigma}_b$, was estimated as $0.047 \pm 0.024$ (mean $\pm$ SD, $n=18$). Since $\hat{\sigma}_p < \hat{\sigma}_b$ the Hybrid-Offset design is likely more cost-effective than a Conventional design (\thmref{1}). Indeed, with $d=5.8\%, \delta = 0.05$, the optimal Hybrid-Offset design is to collect 53 samples and manually annotate 5 (\thmref{2}), for a TSC of 2.3 person-hours per unit; a 68.7\% reduction compared to the Conventional design.

\textbf{Simulation details:} In order to validate the expected cost-savings, simulations were carried out on the coral survey images collected and annotated during 2009. Using the variance estimates from the 2008 data, sample sizes were determined for TSC, $b = 1, 2, \hdots 15$ person-hours using Eq. \eqref{eq:sample_sizes_from_budget}. For each of the 18 units, for each budget, and for $500$ iterations, the required number of images were drawn randomly with replacement from the images pertaining to that unit\footnote{Sampling with replacement was used to avoid finite-population artifacts due to the limited pool of 40 images per unit. In an actual application, there would be a large number of locations from where to capture the images and finite-population correction would not be needed.}, and mean estimates $\hat{\mu}_p$ were calculated. From these estimates the mean error (bias): $\frac{1}{500} \sum_{k=1}^{500} \hat{\mu}_p - \mu_p^0$ and mean absolute error (MAE): $\frac{1}{500} \sum_{k=1}^{500} |\hat{\mu}_p - \mu_p^0|$ of each method was calculated by comparing to the ``ground-truth'' cover $\mu_p^0$, which was estimated from the expert annotations of the $40$ images in the unit. In addition to the Conventional and Hybrid-Offset designs, we include three other designs, which are defined below. All sampling designs are evaluated using both \ftexton and \fCNN.

\textbf{Hybrid-Ratio Design:} Auxiliary information is commonly incorporated using a ratio-estimator which assumes a linear relation between auxiliary values and the values of interest~\cite{sampling, royall1981empirical}. The mean estimate of a ratio-estimator is  $\hat{\mu}_p = \frac{1}{n_b} \sum_{i = 1}^{n_b} \hat{r} f_b(x_i)$, where the ratio, $r$ is estimated as 
\begin{align}
\hat{r}= \frac{\sum_{i=1}^{n_a} f_a(x_i)}{\sum_{i=1}^{n_a} f_b(x_i)} \label{eq:ratio}
\end{align}
The ratio-estimator is design-biased, and there are several approximations of the ratio-estimator variance~\cite{royall1981empirical}. We do not analysis the ratio-estimator, but include it in the simulations using a Hybrid-Ratio design with the same sample sizes as the Hybrid-Offset design \eqref{eq:sample_sizes_from_budget}. We set $\hat{r} = 1$ if the nominator or denominator of \eqref{eq:ratio} is zero in order to avoid ill-defined estimates of $r$.

\textbf{Auxiliary Design:} Using only the auxiliary annotator, $\p$ can be estimated as $\hat{\mu}_p = \frac{1}{n_b} \sum_{i=1}^{n_b} f_b(x_i)$, with sample sizes $n_a = 0, n_b = b/c_c$. This estimator is not design-unbiased, but will have small variance since a large number of samples can be collected and annotated cheaply.

\textbf{Auxiliary Bias-Corrected Design}: The expected value of $f_{b^*}(x)$ for a sample $x$ with value $y \in \{0, 1\}$ is
\begin{align}
  E \left[ {\begin{array}{c}
   f_{b^*}(x) \\
   1 - f_{b^*}(x) \\
  \end{array} } \right] = 
   \left[ {\begin{array}{cc}
   \alpha & 1-\beta \\
   1 - \alpha & \beta \\
  \end{array} } \right]
  \left[ {\begin{array}{c}
   y \\
   1 - y \\
  \end{array} } \right], \label{eq:solownew}
\end{align}
where $\alpha$ and $\beta$ are the classifier sensitivity and specificity respectively. As noted by~\cite{solow2001estimating, hopkins2010method}, an unbiased estimate of $y$ is given by inverting the confusion-matrix in the center of \eqref{eq:solownew} yielding: $\tilde{y} = \frac{f_{b^*}(x) + \beta -1}{\alpha + \beta - 1}$. We denote this operation as ``bias-correction'', since the corrected values are unbiased in expectation. Since $f_b$ is a linear combination of $f_{b^*}$ \eqref{eq:fbstar}, $\p$ can be estimated as $\hat{\mu}_p = \frac{1}{n_b} \sum_{i=1}^{n_b} \frac{f_b(x_i) + \beta -1}{\alpha + \beta - 1}$. Using cross-validation on the data from 2008, specificity and sensitivity were estimated as $\hat{\alpha} = 0.738, \hat{\beta} = 0.963$ for $f_{b^*}^{(\textrm{texton})}$, and $\hat{\alpha} = 0.740, \hat{\beta} = 0.968$ for $f_{b^*}^{(\textrm{CNN})}$, and used for the bias-correction. An analysis of the variance of a mean estimator from abundance-corrected values is provided in Appendices~\ref{sec:transfer-sampling} \& \ref{sec:twostage}.

\begin{figure}[t]
   \begin{center}
    \includegraphics[width=0.95\linewidth]{\figuredir 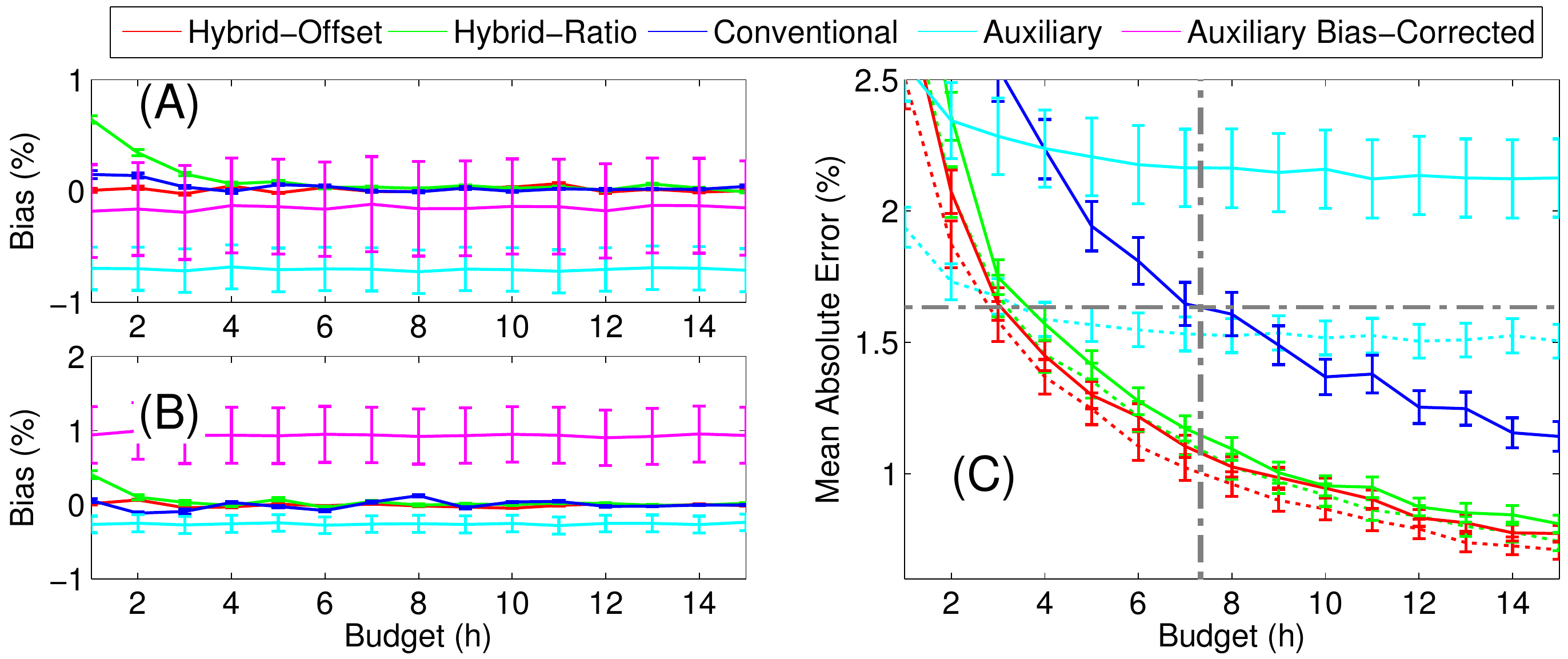}
    \end{center} \vspace{-4mm}
    \caption{\textbf{Simulations results on the coral reef survey data.} Results displayed as mean $\pm$ SE for $n=18$ sampling units and for TSC (budgets) between 1 and 15 person-hours per unit. Estimator bias is shown for the \ftexton auxiliary annotator in panel \textbf{(A)}, and for \fCNN in panel \textbf{(B)}. Mean Absolute Error (MAE) is shown in panel \textbf{(C)}, where estimates using \ftexton are indicated with solid lines, and \fCNN with dotted. Gray dash-dotted lines indicate the current operating point of the Conventional sampling design currently utilized by the ecological monitoring program. The MAEs of the Bias-Corrected estimates were $>3\%$ (shown in \figref{supp_results}).} \vspace{-2mm}
   \label{fig:simulations}
\end{figure}

\section{Results \& Discussion}
As expected, the Conventional and Hybrid-Offset designs were unbiased (\figref{simulations}A). The Hybrid-Ratio design also had low bias, except for smaller budgets, where, as mentioned above, the ratio-estimator \eqref{eq:ratio} may be ill-posed. The Auxiliary Bias-Corrected design was, in-fact, biased (\figref{simulations}A, B). This may seem surprising, but the corrected estimates are only unbiased if the sensitivity, $\alpha$ and specitivity, $\beta$ are known~\cite{solow2001estimating, hopkins2010method}. These results thus indicate that $\hat{\alpha}$ and $\hat{\beta}$ estimated from the 2008 data were not valid for the sampling units from 2009. This may be due to domain-shifts, which can severely affect the performance of machine-learning based classifieres~\cite{torralba2011unbiased, pan2010survey}. One way to circumvent this problem is to use Bias-Correction in a Hybrid design, and estimate $\hat{\alpha}$ and $\hat{\beta}$ from the $n_a$ samples annotated by both annotators. However, as shown in \appref{transfer_comparison}, such Hybrid-Bias-Correction design is inferior to the Hybrid-Offset design. Finally, and less surprisingly, the un-corrected Auxiliary design was biased, although to a lesser extent for \fCNN (\figref{simulations}A, B).

The MAE of the Conventional design at the $7.3$ person-hour budget currently utilized by the MCR-LTER program was $1.63 \pm .32\%$ (mean $\pm$ SE, $n = 18$; \figref{simulations}B). This was outperformed by the Hybrid-Offset and Hybrid-Ratio designs, which utilized both annotators. At the $7.3$ person-hour budget, the Hybrid-Offset estimator MAE was $1.08 \pm .16\%$ when relying on \ftexton and $1.00 \pm .18\%$ when relying on \fCNN, which is significantly lower than for the Conventional design. Conversely, the MAE that the Conventional design achieved at $7.3$ person-hours, can be achieved by the Hybrid-Offset design at around $3.5$ person-hours; a $\sim 50\%$ cost reduction. The Hybrid-Ratio design, while comfortably better than the Conventional design, performed worse than the Hybrid-Offset design for all budgets. This may be becuase the sample sizes were optimate for the offset-estimator and not the ratio-estimator. The Auxilary and Auxiliary Bias-Corrected designs, which both relied only on the auxiliary annotator, performed weaker than the Hybrid and Conventional designs. The MAE of the Bias-Corrected design was $>3\%$ for all budgets, indicating that the correction method of \cite{solow2001estimating, hopkins2010method} was ineffective. The Auxiliary design performed poorly when relying on \ftexton (MAE $>2\%$), but better when relying on \fCNN, barely outperforming the Conventional design at the $7.3$ person-hour budget (\figref{simulations}C).

It is also clear from the simulations that \fCNN outperformed \ftexton, reducing the errors for the Hybrid and Auxiliary estimators. This is expected as CCN-based methods have recently achieved state of the art performance on several visual recognition tasks~\cite{girshick2014rich}. As new, and stronger classification methods are developed, the requirement of \thmref{1} will be satisfied for an increasing number of applications, suggesting increasing utility of the offset sampling design.

We have used linear cost-functions throughout this work, with a fixed cost per sample. In reality, the cost per sample is likely to decrease when more samples are collected. This is true, in particular for spatial surveys where the sample collector will, on average, have shorter distance to travel between the samples. Since Hybrid sampling designs require larger number of collected samples, the cost-savings estimated by our simulations should be considered a lower-bound on the actual cost-savings.

\section{Conclusion}
We have investigated the implications of modeling and incorporating the cost and accuracy of two annotators in random sampling designs for population mean estimation, and shown that significant cost-savings are possible using data from a coral reef survey. However, the derived formulations are general and applicable in many other situations. These includes other marine surveys of e.g. fish~\cite{willis2000baited} or plankton~\cite{olson2007submersible}, terrestrial surveys of crops~\cite{hunt2010acquisition}, forests~\cite{getzin2012assessing}, rangelands~\cite{rango2009unmanned} and deserts~\cite{archer1988autogenic}; and audio-based surveys of e.g. marine mammal or bird populations~\cite{johnson2003digital}. 

To the best of our knowledge, this is the first work that models automated annotation as part of the sampling design and we believe that there are several interesting directions of future work, notably with respect to stratified and sequential sampling procedures.

\bibliographystyle{abbrv}
\bibliography{vs_nips}

\pagebreak
\begin{widetext}
\begin{center}
\textbf{\Large Appendices: Random Sampling in an Age of Automation: Minimizing Expenditures through Balanced Collection and Annotation}
\end{center}
\setcounter{equation}{0}
\setcounter{figure}{0}
\setcounter{table}{0}
\setcounter{page}{1}
\makeatletter
\renewcommand{\theequation}{S\arabic{equation}}
\renewcommand{\thefigure}{S\arabic{figure}}
\renewcommand{\bibnumfmt}[1]{[#1]}
\renewcommand{\citenumfont}[1]{#1}
\appendix
\section{Variance of the offset estimator}
\label{sec:offset_variance}
To derive the variance of $\hp$ we begin by expanding out and separating all terms,
\begin{align}
\hp &=  \frac{1}{\nb} \sum_{i = 1}^\nb f_b(x_i) - \hat{\mu}_b. \\
& = \frac{1}{\nb} \sum_{i = 1}^\nb f_b(x_i) - \frac{1}{\na} \sum_{i = 1}^\na f_b(x_i) - f_a(x_i) \\
& = \frac{1}{\nb} \sum_{i = 1}^\nb y_i + \frac{\na  -\nb}{\na \nb} \sum_{i = 1}^\na \epsilon_{b,i} + \frac{1}{\nb}  \sum_{i = \na+1}^\nb \epsilon_{b,i} + \frac{1}{\na}  \sum_{i = 1}^\na \epsilon_{a,i},
\end{align}
after which the variance of $\hp$ is given by:
\begin{align}
\var[\hp] =  &\frac{1}{\nb} (\sigma_p^2 - \sigma_b^2 ) + \frac{1}{\na} (\sigma_a^2  + \sigma_b^2) + \frac{2}{\nb}\sigma_{a, p} + 2 \frac{\na - \nb}{\na \nb} \sigma_{a, b}.
\end{align}
The first two terms are the variance of the data and annotator errors. The third term is the covariance between the primary annotator errors and the sample values $\sigma_{a,p}$. Since, by assumption, the primary annotator is ``accurate'', we can expect this term to be small, and can be omitted. The last term accounts for the covariance between the annotator errors $\sigma_{a, b}$. Since the two annotators are operating independently, this is assumed to be zero, and it can also be omitted. Interestingly, the third covariance term, $\sigma_{p, b}$, between the auxiliary annotator and the actual values, which may be significant, cancels out and does not affect the final expression.

\section{Proof of Theorem 1}
\begin{proof}
It follows from \eqref{eq:pvar_breve} and \eqref{eq:pvar} that if $\sigma_b = \sigma_p \Rightarrow \var[\hp] = \var[\bp]$. Since $\var[\hp]$ increases monotonically with increasing $\sigma_b$ (since $\nb > \na$) the theorem follows trivially. 
\end{proof}

\section{Proof of Theorem 2}
\label{sec:tcd_offset_convex}
We start with the following lemma:
\begin{Lemma}
The cost function of offset sampling \eqref{eq:tcd_offset} is convex for $\sigma_b^2 < \sigma_p^2$, $\nb \geq \ntrad$ where $\ntrad$ is given by \eqref{eq:ntrad}.
\end{Lemma}

\begin{proof}
We begin by deriving the first and second derivatives of \eqref{eq:tcd_offset}:
\begin{align}
\frac{\partial t(\na, \nb)}{\partial \nb} &= (\cc + \cb)\left(1 - \frac{k(\sigma_b^2 + \sigma_a^2)(\sigma_p^2 - \sigma_b^2)}{(\nb \frac{d^2}{\zeta_{\delta}^2} - (\sigma_p^2 - \sigma_b^2))^2} \right)
\end{align}
\begin{align}
\frac{\partial^2 t(\na, \nb)}{\partial \nb^2} &= (\cc + \cb) \left( \frac{ 2k(\sigma_b^2 + \sigma_a^2)(\sigma_p^2 - \sigma_b^2)(\frac{d^2}{\zeta_{\delta}^2})}{(\nb \frac{d^2}{\zeta_{\delta}^2} - (\sigma_p^2 - \sigma_b^2))^3} \right).
\end{align}
Since by assumption $\sigma_b^2 < \sigma_p^2$, the second derivative is positive for 
\begin{align}
\nb \geq \ntrad > \frac{\zeta_{\delta}^2}{d^2} (\sigma_p^2 - \sigma_b^2),
\end{align}
where the strict inequality requires either $\sigma_a^2$ or $\sigma_b^2$ to be non-zero. This concludes the proof.
\end{proof}
Since, according to Theorem~\ref{theorem:1}, offset sampling should only be considered if $\sigma_b^2 < \sigma_p^2$, and since by design, $\nb \geq \ntrad$, \eqref{eq:tcd_offset} is convex, and Theorem~\ref{theorem:2} follows by setting the first derivative to zero and solving for $\nb$.

\section{Proof of Theorem 3}
\label{sec:kzero}
\begin{proof}
The TSC of conventional sampling, under the assumption that $c_b = 0$, is given by $\na(\cc + \ca)$ and the TSC of offset sampling is given by $\nb \cc + \na \ca$. As noted previously, if $\nhyb = \hhyb = \ntrad$, offset sampling reduces to conventional sampling and the TSCs are equal. Since $\nhyb$ minimize the offset TSC, which is convex for $\nb \geq \ntrad$ (Appendix \ref{sec:tcd_offset_convex}), it follows that the TSC of offset sampling is smaller than the TSC of conventional sampling if and only if $\nhyb > \ntrad$. The threshold, $\sigma_{\Delta}$ for when this occurs can be calculated by equating the two arguments inside the max operator of \eqref{eq:nopt} and solving for $k$:

\begin{align}
\frac{\zeta_{\delta}^2}{d^2} (\sigma_p^2 + \sigma_a^2) &= \frac{\zeta_{\delta}^2}{d^2} \left( \sigma_p^2 - \sigma_b^2  + \sqrt{\sigma_{\Delta} (\sigma_b^2+ \sigma_a^2) ( \sigma_p^2 - \sigma_b^2 )}~ \right) \Rightarrow \\
\sigma_a^2 + \sigma_b^2 &=  \sqrt{\sigma_{\Delta} (\sigma_b^2+ \sigma_a^2) ( \sigma_p^2 - \sigma_b^2 )} \Rightarrow \\
\sigma_{\Delta} &= \frac{\sigma_a^2 + \sigma_b^2}{\sigma_p^2 - \sigma_b^2}.
\end{align}
\end{proof}

\section{Random Sampling with Abundance Correction}
\label{sec:transfer-sampling}
A third sampling design can be defined under one critical additional assumption. This assumption, which we will denote the `population drift assumption', is that the performance of $f_b$ can be defined in terms of a matrix of confusion, $Q$ (which excludes real-valued output spaces ${\cal Y}$), and that $Q$ is known \emph{a priori} for the data to be sampled.

For these derivations we will let ${\cal Y} = \{0, 1\}$, meaning that the samples $y$ are drawn from $\ber(\p)$, a Bernoulli distribution with mean $\p$. This corresponds to annotating each sample $x_i$ as containing or not containing the quantity of interest. In \appref{twostage} we derive the statistics of auxiliary sampling in a two-stage sampling design~\cite{sampling, deming1966some}, where each sample $x_i$ is annotated by second stage sampling, from which the corresponding $y_i$ is obtained.

A matrix of confusion, $Q$ characterizes the misclassification rates of $f_b$. In the binary case, $Q$ is a two by two matrix
\[
   Q=
  \left[ {\begin{array}{cc}
   \alpha & 1-\beta \\
   1 - \alpha & \beta \\
  \end{array} } \right],
\]
where $\alpha$ is the sensitivity and $\beta$ the specificity.  As noted independently by~\cite{solow2001estimating, hopkins2010method} $Q$ can be used to create an unbiased estimate of $y$, and we begin by recalling this procedure. The expected value of the auxiliary annotation $f_b$ of a sample $x$ with value $y$ is
\[
  E \left[ {\begin{array}{c}
   f_b(x) \\
   1 - f_b(x) \\
  \end{array} } \right] = 
  Q
  \left[ {\begin{array}{c}
   y \\
   1 - y \\
  \end{array} } \right],
\]
and an unbiased estimator of $y$ is given by inverting $Q$
\begin{align}
\tilde{y} = \frac{f_b(x) + \beta -1}{\alpha + \beta - 1}. \label{eq:solow}
\end{align}
We refer to this as the `abundance corrected' value, and derive a sampling procedure based on this correction. The variance of $\tilde{y}_i$, given the true value $y_i$ is 
\begin{align}
\var[\tilde{y}_i | y_i] = \frac{\var[f_b(x_i) | y_i]}{(\alpha + \beta - 1)^2}, \label{eq:varyhat}
\end{align}
which follows directly from \eqref{eq:solow}. We also note that
\begin{align}
\var[f_b(x_i) | y_i] = y_i \alpha (1-\alpha) + (1-y_i) (1-\beta) \beta, \label{eq:varf}
\end{align}
since if $y_i = 1, f_b(x_i) \sim \ber(\alpha)$, and if  $y_i = 0, f_b(x_i) \sim \ber(1-\beta)$.
Combining \eqref{eq:varyhat} and \eqref{eq:varf}, yields
\begin{align}
\var[\tilde{y}_i | y_i] = \frac{y_i \alpha (1-\alpha) + (1-y_i) (1-\beta) \beta}{(\alpha + \beta - 1)^2}.
\end{align}
Finally, $\var[\tilde{y}_i]$ is given by the law of total variance
\begin{align}
\var[\tilde{y}_i] = & E\left[ \var(\tilde{y_i} | y_i) \right] + \var\left[E(\tilde{y}_i | y_i)\right] \\
		= & \sigma_s^2 + \sigma_p^2,
\end{align}
where $\sigma_p^2$ is the data variance and
\begin{align}
\sigma_s^2 = \frac{\p \alpha (1-\alpha) + (1-\p) (1-\beta) \beta}{(\alpha + \beta - 1)^2}, \label{eq:vars}
\end{align}
the variance introduced by the abundance correction. If the classifier is balanced, i.e. $\alpha = \beta$, \eqref{eq:vars} simplifies to $ \sigma_s^2 = \frac{\alpha (1-\alpha)}{(2\alpha - 1)^2}$. Since $\tilde{y}_i$ is an unbiased estimator of $y_i$, we can achieve an unbiased estimation of $\p$ as
\begin{align}
\tp = \frac{1}{\nb} \sum_{i=1}^\nb \tilde{y}_i, \label{eq:psolow}
\end{align}
with variance, assuming that $\sigma_s^2$ and $\sigma_p^2$ are uncorrelated
\begin{align}
\var[\tp] = \frac{1}{\nb}(\sigma_s^2 + \sigma_p^2),
\end{align}
and sample size
\begin{align}
\ntrans = \frac{\zeta_{\delta}^2}{d^2} (\sigma_s^2 + \sigma_p^2). \label{eq:nsolow}
\end{align}
Finally, the TSC, since $\na = 0$, is given by
\begin{align}
t(\na, \nb) = (\cc + c_b)\ntrans. \label{eq:tcd_transfer}
\end{align}
The auxiliary sampling design requires annotation of $\ntrans$ samples by the auxiliary annotator, but as it does not require any annotations by the primary annotator, the TSC can be low. The following theorem is given directly from \eqref{eq:tcd_transfer} and the cost function for conventional sampling.
\begin{Theorem}
For binary output spaces, ${\cal Y} = \{0, 1\}$, the TSC of auxiliary sampling is smaller than the TSC of conventional sampling if and only if
\begin{align}
k' > \frac{\sigma_p^2 + \sigma_s^2}{\sigma_p^2 + \sigma_a^2} \label{eq:kprime}
\end{align}
where $k' = (\cc + c_a)/(\cc + c_b)$.
\end{Theorem}
If $f_b$ is accurate then $\sigma_s^2$ is small and auxiliary sampling is cheaper then conventional sampling even for low primary annotation costs. For example, if $\alpha = \beta = 0.9 \Rightarrow \sigma_s^2 \approx 0.14$,  $\sigma_p^2 = 0.04$, and $\sigma_a^2 = 0$, it suffices that $k'$ is larger than $4.5$, which is satisfied e.g. if $c_b = 0$ and $c_a>3.5c_c$.  If, on the other hand, $\alpha = \beta = 0.7 \Rightarrow \sigma_s^2 \approx 1.3$, $k'$ must be larger than $33.5$.

\section{Two-Stage Random Sampling with Abundance Correction}
\label{sec:twostage}
In two-stage sampling designs each first stage sample, $x_i$ is again sampled randomly using some number, $s$ of second-stage samples~\cite{sampling}. An analysis of the errors using such designs is provided by Deming~\cite{deming1966some}. Second stage sampling is commonly used e.g. in benthic surveys where each collected photoquadrat is annotated using random point sampling~\cite{kohler2006coral}. This protocol requires $s$ points to be overlaid on each image at locations selected randomly with replacement. The substrate under each point is then annotated by an expert as pertaining to one of some number of classes. An unbiased estimator of the abundance (benthic cover) of each class for a certain sample can be derived by counting how many of the $s$ annotations that were annotated as that class. 

We derive the statistics of two-stage sampling under the population drift assumption, namely that each decision is made by some noisy annotator $f_b$, with known matrix of confusion. We will denote by $x_{i1} \hdotsc x_{is}$ the $s$ locations to be annotated in each sample $x_i$, and $u_{ij} \in \{0, 1\}$ the true value associated with each location. The value of each first stage sample is approximated by $y_i = \sum_{j = 1}^s u_{ij}$. We do not make any assumptions on the probability density from which the first stage samples are drawn, but as previously let $\p$ denote the expected value and $\sigma_p^2$ the variance. Given a classifier $f_b$ with known matrix of confusion, an unbiased estimator of ${u}_{ij}$ is given as previously by
\begin{align}
\tilde{u}_{ij} = \frac{f_b(x_{ij}) + \beta - 1}{\alpha + \beta - 1}. \label{eq:total0}
\end{align}
From this an unbiased estimator of $\tilde{y}_i$ is given by
\begin{align}
\tilde{y}_i = \frac{1}{s} \sum_{j = 1}^s \tilde{u}_{ij}.
\end{align}
We have derived the variance of $\tilde{y}_i$ for the special case where $s=1$ in the main paper. Next, we show how to derive the variance of $\tilde{y}_i$ for a general $s$ by applying the law of total variation twice. We begin by noting that
\begin{align}
\var(\tilde{y}_i) = E[\var(\tilde{y}_i | y_i)] + \var[E(\tilde{y}_i | y_i)]. \label{eq:total1}
\end{align}
The second term is simply $\var[E(\tilde{y}_i | y_i)] = \var[y_i] = \sigma_p^2$, and the first term can be expressed in terms of $\tilde{u}_{ij}$
\begin{align}
\var (\tilde{y}_i | y_i) = \frac{1}{s^2} \sum_{i=1}^s \var(\tilde{u}_{ij} | y_i),
\end{align}
which can be expressed, by again using the law of total variation, as
\begin{align}
\var(\tilde{u}_{ij}|y_i) &= E[\var(\tilde{u}_{ij} | u_{ij}, y_i)] + \var[ E ( \tilde{u}_{ij} | u_{ij}, y_i) )]. \label{eq:total15}
\end{align}
The second term of \eqref{eq:total15} is simply given by $\var[ E ( \tilde{u}_{ij} | u_{ij}, y_i) )] = \var[ u_{ij} | y_i] =y_i (1-y_i)$, but the first term is less obvious. Following Solow et al.~\cite{solow2001estimating}, we first note that 
\begin{align}
\var(f_b(x_{ij}) | u_{ij}, y_i) = u_{ij} \alpha(1 - \alpha) + (1-u_{ij})\beta(1-\beta),
\end{align}
since if $u_{ij} = 1, f_b(x_{ij}) \sim \ber(\alpha)$, and if  $u_{ij} = 0, f_b(x_{ij}) \sim \ber(1-\beta)$.
We then note that
\begin{align}
\var[\tilde{u}_{ij}|u_{ij}, y_i] &= \frac{ \var[f_b(x_{ij})|u_{ij}, y_i]}{(\alpha + \beta - 1)^2},
\end{align}
which follows directly from \eqref{eq:total0}, and also that $E[u_{ij}] = y_i$. Putting this together yields the following expression for the first term of \eqref{eq:total15}:
\begin{align}
E[\var(\tilde{u}_{ij} | u_{ij}, y_i)] & = \frac{y_i\alpha(1-\alpha) + (1-y_i)\beta(1-\beta)}{(\alpha + \beta -1)^2}.
\end{align}
Putting this all together yields
\begin{align}
& \var(\tilde{y}_i) = E[\var(\tilde{y}_i | y_i)] + \var[E(\tilde{y}_i | y_i)] \\
=& E[\frac{1}{s^2} \sum_{i=1}^s \var(\tilde{u}_{ij} | y_i)] + \sigma_p^2 \\
=& \frac{1}{s^2} \sum_{i=1}^s E\left[\var(\tilde{u}_{ij} | y_i) \right] + \sigma_p^2 \\
=& \frac{1}{s^2} \sum_{i=1}^s E\left[\frac{y_i\alpha(1-\alpha) + (1-y_i)\beta(1-\beta)}{(\alpha + \beta -1)^2} + y_i (1-y_i) \right] + \sigma_p^2 \\
=& \frac{1}{s^2} \sum_{i=1}^s \left[ \frac{\p\alpha(1-\alpha) + (1-\p)\beta(1-\beta)}{(\alpha + \beta -1)^2} + \p (1-\p) \right] + \sigma_p^2 \\
=& \frac{1}{s} [\sigma_s^2 + \p (1-\p)] + \sigma_p^2,
\end{align}
where $\sigma_s^2$ is given by \eqref{eq:vars}. Interestingly, the variance of $\tilde{y}$ approach $\sigma_p$ for large number of secondary stage samples, $s$. This is to be expected under the assumption that $f_b$ is perfectly modeled by a known matrix of confusion $Q$. Since $\tilde{y}_i$ is an average across $s$ decisions, the variance introduced by the abundance correction cancels out with large values of $s$. 

Finally, the total variance of $\tp$ is given by
\begin{align}
\var[\tp] =& \frac{1}{\nb} \left( \frac{1}{s} [\sigma_s^2 + \p (1-\p)] + \sigma_p^2 \right),
\end{align}
and the sample size by
\begin{align}
\nb = \frac{\zeta_{\delta}^2}{d^2} \left[ \frac{1}{s} [\sigma_s^2 + \p (1-\p)] + \sigma_p^2 \right]. \label{eq:nsolow}
\end{align}

\section{Supplementary results}
\label{sec:supp_results}
Detailed simulation results are shown in \figref{supp_results}.

\begin{figure}[t]
   \begin{center}
    \includegraphics[width=0.85\linewidth]{\figuredir 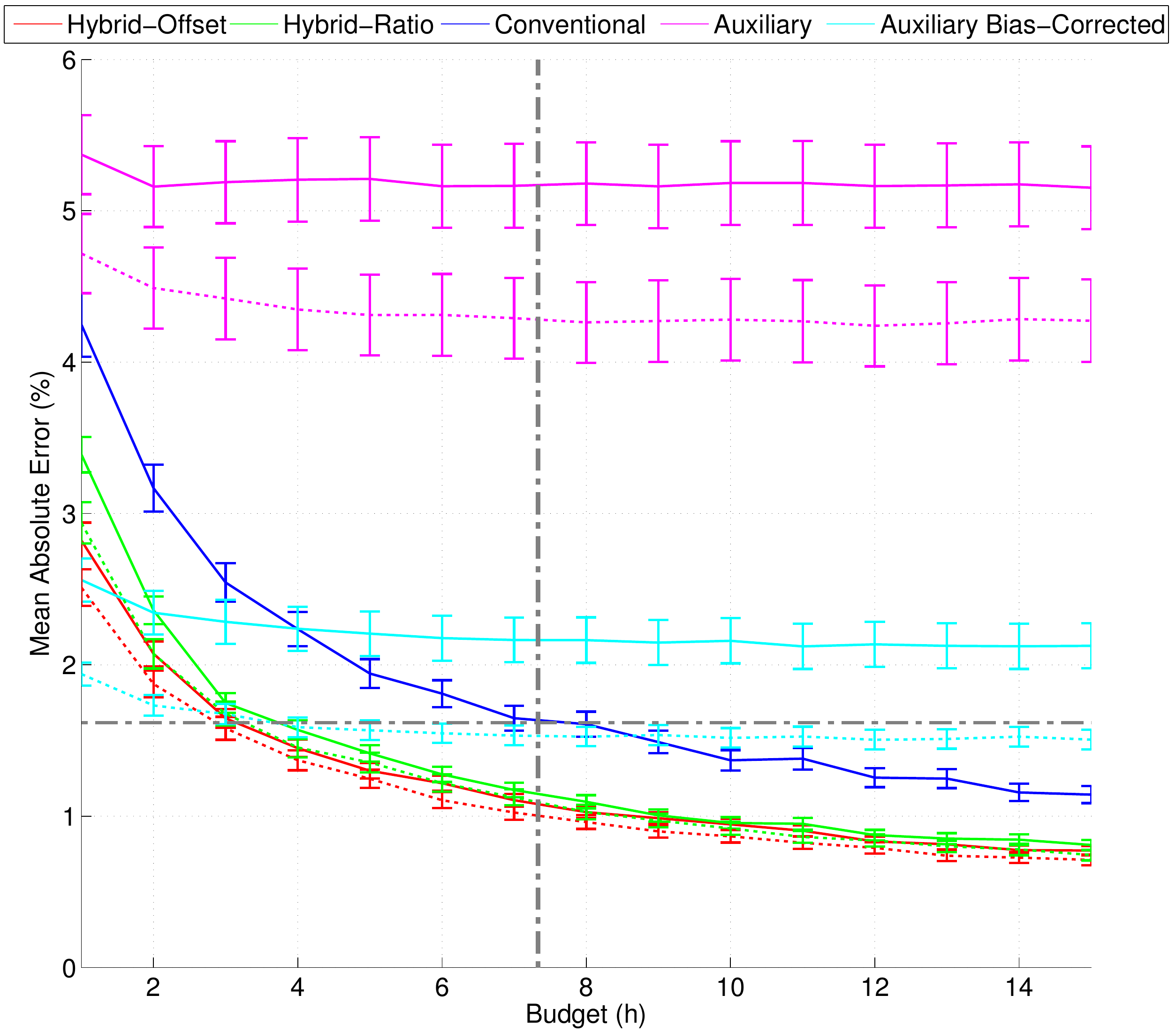}
    \includegraphics[width=0.85\linewidth]{\figuredir 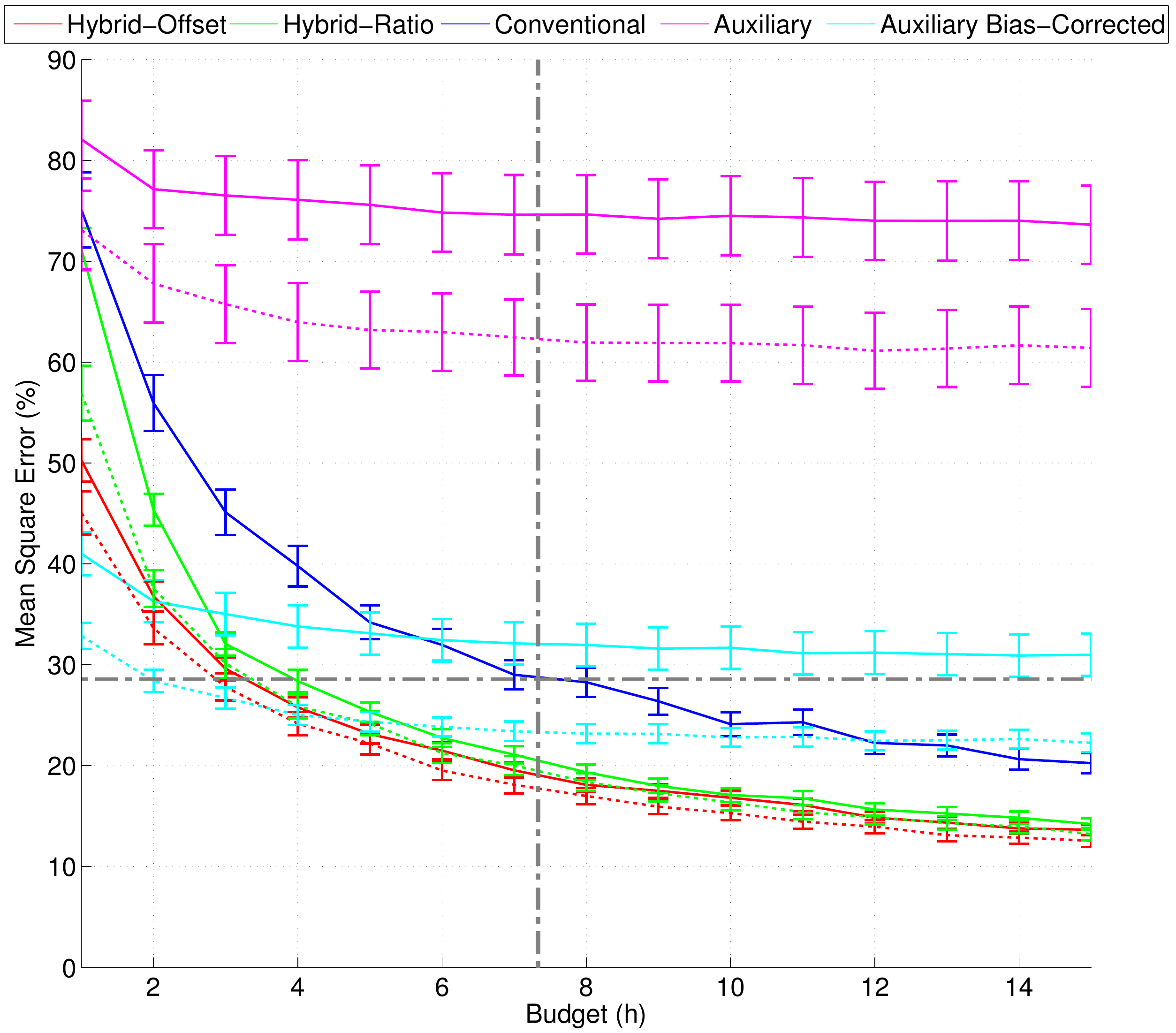}
    \end{center}
    \caption{\textbf{Supplementary Simulations Restuls} Results displayed as mean $\pm$ SE for $n=18$ sampling units, for TSC (budgets) between 1 and 15 person-hours per unit and for (Top) Mean Absolute Error and (Bottom) Mean Square Error} 
   \label{fig:supp_results}
\end{figure}

\section{Bias-Correction Sampling With Unknown Confusion Matrix}
\label{sec:transfer_comparison}

The Auxiliary Bias-Correction design evalued in the simulations assumes that the specificity and sensitivity of $f_b$ is known \emph{a priori} for the data to be sampled. This assumption is strong, and may not always hold. In such cases, one could rely on a Hybrid sampling design and use $\na$ samples annotated by both annotators to estimate $\hat{\alpha}$ and $\hat{\beta}$. In such Hybrid-Bias-Correction design, $\alpha$ and $\beta$ can be estimated as
\begin{align}
\hat{\alpha} &= \frac{\sum_{i = 1}^\na f_{b^*}(x_i) f_{a^*}(x_i)}{ \sum_{i = 1}^\na f_{a^*}(x_i)} \label{eq:alpha}\\
\hat{\beta} &= \frac{\sum_{i = 1}^\na (1 - f_{b^*}(x_i)) (1 - f_{a^*}(x_i))}{ \sum_{i = 1}^\na (1-f_{a^*}(x_i))}, \label{eq:beta}
\end{align}
and an estimator of $\p$ can be defined as
\begin{align}
\hat{\mu}_p = & \frac{1}{\nb} \left[ \sum_{i=1}^\na f_{a}(x_i) + \sum_{i=\na+1}^\nb \frac{f_{b}(x_i) + \hat{\beta} -1}{\hat{\alpha} + \hat{\beta} - 1} \right]. \label{eq:pbiased}
\end{align}

However, we argue that such design is inferior to hybrid sampling for several reasons. First, the bias-corrected mean estimate of \eqref{eq:pbiased} is biased if estimates of $\alpha$ and $\beta$ are used in place of the true values~\cite{solow2001estimating}. Second, it is difficult to derive an analytical expression for $\var[\hat{\mu}_p]$ that accounts for the variances of $\hat{\alpha}$ and $\hat{\beta}$. Without this expression, one cannot derive optimal sample sizes. Third, simulations detailed below indicate that the Hybrid-Offset design achieved lower errors for the same TSC for a wide array of parameters ($\alpha, \beta, \p, \na, \nb$). 

\textbf{Simulations:} For all combinations of $\alpha = [0.6, 0.8, 0.95], \beta = [0.6, 0.8, 0.95], \p = [0.5, 0.75, 0.9]$, $\na = [100, 150 \hdotsc 500]$, and $\nb=1000$, the following simulation was performed. First, $\na$ samples $f_a(x_1) \hdotsc f_a(x_{\na})$ were drawn from a Bernoulli (Ber) distribution with mean $\p$. For each $f_a(x_i) = 1$, $f_b(x_i)$ was drawn from $\textrm{Ber}(\alpha)$, and for each $f_a(x_i) = 0$, $f_b(x_i)$ was drawn from $\textrm{Ber}(1 - \beta)$. The parameters $\alpha, \beta$ and $\mu_b$ were then estimated according to \eqref{eq:alpha}, \eqref{eq:beta}, and \eqref{eq:bhat}. Finally $\nb-\na$ new samples $f_b(x_i)$ were drawn using the same procedure and used to estimes ${\mu}_p$ from \eqref{eq:pbiased} and from \eqref{eq:phat}. This procedure was repeated 2000 times and sample standard deviations were calculated. The signed difference between the standard deviations were calculated for each value of $\alpha$, $\beta$, $\p$, and $\na$, is shown in \figref{errordiff}. These results indicate that the Hybrid-Offset design is more accurate that the Hybrid-Bias-Corrected design for all parameters.

\begin{figure*}[ht]
   \begin{center}
    \includegraphics[width=.95\linewidth]{\figuredir 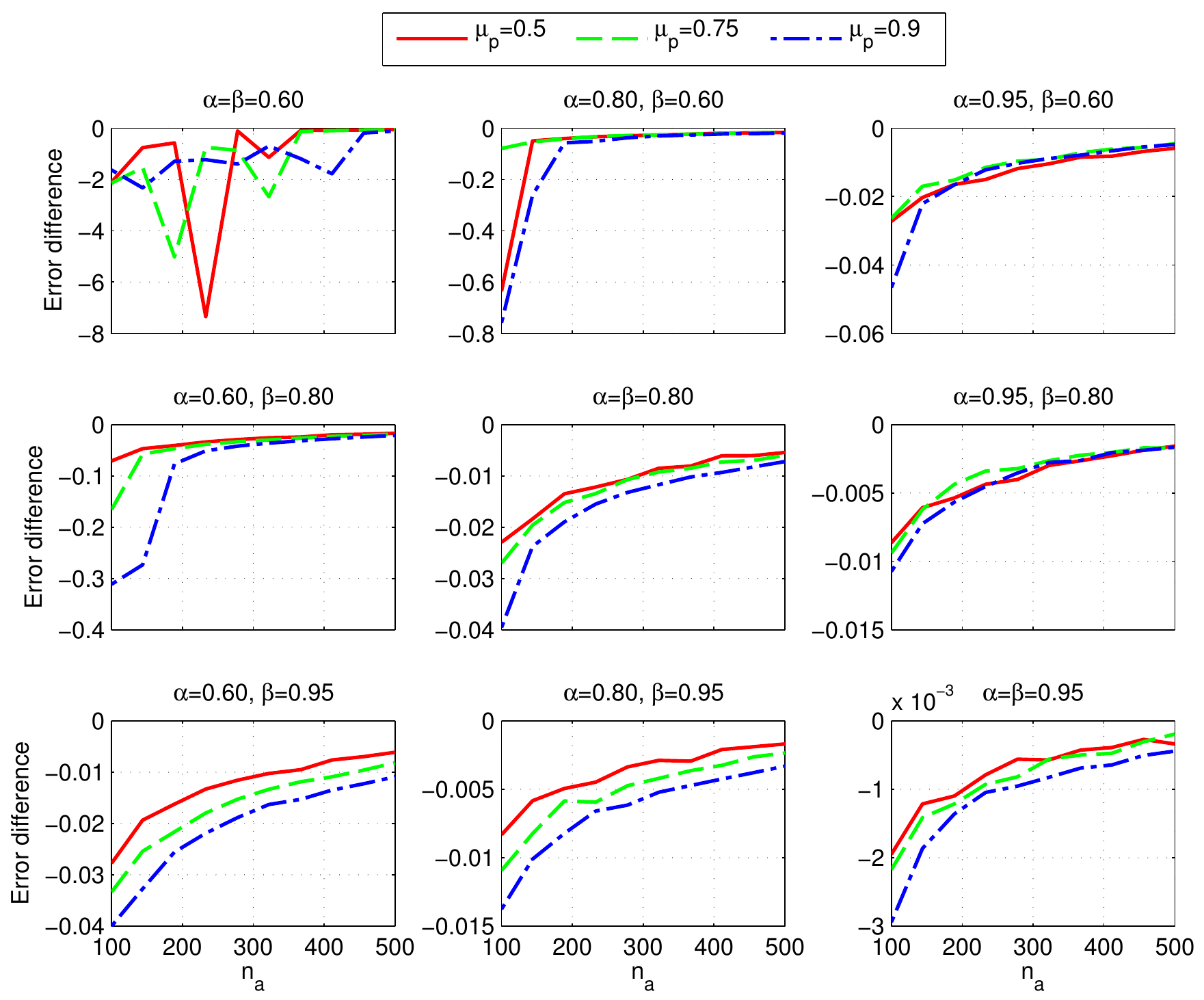}
    \end{center}
    \caption{Difference in estimated standard errors between the Hybrid-Offset design and the Hybrid-Bias-Corrected design, for $\nb=1000$, and different values of $\na$, $\alpha$, $\beta$, and $\p$. Note that all differences are negative, indicating that the sampling errors of the Hybrid-Offset design are smaller}
   \label{fig:errordiff}
\end{figure*}

\end{widetext}
\end{document}